\documentclass[pra,onecolumn,superscriptaddress,nofootinbib]{revtex4-2}

\usepackage{graphics}
\usepackage{tikz}
\usetikzlibrary{shapes.geometric}
\usepackage{amsfonts,amscd,amsmath,amsthm}
\usepackage{enumerate}
\usepackage{epsfig}
\usepackage{subfigure}
\usepackage{xcolor}
\usepackage[colorlinks = true]{hyperref}
\usepackage{physics}
\usepackage{epstopdf}
\usepackage{framed}
\usepackage{multirow}
\usepackage{color}
\usepackage{longtable}
\usepackage{comment}
\usepackage[ruled,vlined]{algorithm2e}
\SetKwInput{KwInput}{Input}
\SetKwInput{KwOutput}{Output}

\graphicspath{{./figure/}}

\newtheorem{observation}{Observation}
\newtheorem{definition}{Definition}
\newtheorem{Proposition}{Proposition}

\begin{document}

\preprint{APS/123-QED}
\title{Efficient entanglement generation and detection of generalized stabilizer states}

\author{Yihong Zhang}
\affiliation{Center for Quantum Information, Institute for Interdisciplinary Information Sciences, Tsinghua University, Beijing 100084, China}
\author{Yifan Tang}
\affiliation{Center for Quantum Information, Institute for Interdisciplinary Information Sciences, Tsinghua University, Beijing 100084, China}
\author{You Zhou}
\email{zyqphy@gmail.com}
\affiliation{School of Physical and Mathematical Sciences, Nanyang Technological University, 637371, Singapore}
\affiliation{Zhongguancun Haihua Institute for Frontier Information Technology, Beijing 100084, China}
\author{Xiongfeng Ma}
\email{xma@tsinghua.edu.cn}
\affiliation{Center for Quantum Information, Institute for Interdisciplinary Information Sciences, Tsinghua University, Beijing 100084, China}

\begin{abstract}
The generation and verification of large-scale entanglement are essential to the development of quantum technologies. In this paper, we present an efficient scheme to generate genuine multipartite entanglement of a large number of qubits, by using the Heisenberg interaction. This method can be conveniently implemented in various physical platforms, including superconducting, trapped-ion, and cold-atom systems. In order to characterize the entanglement of the output quantum state, we generalize the stabilizer formalism and develop an entanglement witness method. In particular, we design a generic searching algorithm to optimize entanglement witness with a minimal number of measurement settings under a given noise level. From the perspective of practical applications, we numerically study the trade-off between the experiment efficiency and the detection robustness.
\end{abstract}


\maketitle
\section{Introduction}
Entanglement is a crucial resource in quantum mechanics and plays a central role in  quantum information processing, including quantum communication \cite{Bennett1993Teleporting}, quantum computing \cite{raussendorf2001one, Galindo2002rmp}, quantum cryptography \cite{Xu2020Secure}, and quantum metrology \cite{giovannetti2011advances}. The experimental research of quantum information is committed to realizing massive quantum equipment, such as large-scale quantum networks \cite{cirac1997quantum} and fault-tolerant quantum computers \cite{shor1996fault, Nigg2014QEC}. One of the key objectives in this field is to generate entanglement among a large number of degree of freedoms, known as genuine multipartite entanglement (GME) \cite{Horodecki2009entanglement}. The size of the genuinely entangled quantum system becomes a figure of merit for assessing the advancement of quantum devices in the competition among various realizations.

The generation of large-scale GME generally requires a number of entangling gates. In the noisy intermediate-scale quantum era \cite{preskill2018quantum}, it is challenging to operate quantum gates precisely with a deep circuit, due to the decoherence and noise, which becomes the main obstacle in generating large-scale quantum entanglement. Compared with single-qubit gates, two-qubit entangling gates would normally introduce more noises to the system. Thus, in a quantum information processing task, we should minimize the number of two-qubit gates to suppress the negative effects of the noise. In this paper, we propose an efficient scheme based on the Heisenberg interaction to generate GME in large-scale quantum systems, which only contains a very shallow circuit and can be directly applied on many experimental platforms.

Due to the exponential scaling of the Hilbert space with the number of subsystems, it is generally a challenging task to fully detect the entanglement of a large quantum system, which normally requires quantum state tomography. Numerous theoretical and experimental efforts have been devoted to characterize GME \cite{Friis2019Reviews},
such as ion trap \cite{Monz2011,Friis2018}, photons \cite{Wang2018,Zhong2018}, Rydberg atoms \cite{Omran2019} and superconducting circuits \cite{Gong2019,Wei2020,Song2019}. With certain preknowledge of the prepared state, the entanglement witness (EW) \cite{TERHAL2001witness,guhne2009entanglement} has been proposed to efficiently determine whether a state is entangled by using a limited number of measurements. In general, the EW is a Hermitian operator $\mathcal{W}$ which has non-negative expectation values for all separable states and negative values for some specific entangled states \cite{guhne2009entanglement}, i.e.,
\begin{equation}\label{eq:EW_criterion}
\begin{split}
&\langle \mathcal{W}\rangle \geq 0, \mathrm{for\ all\ separable\ states}; \\
&\langle \mathcal{W}\rangle < 0, \mathrm{for\ some\ entangled\ states}.
\end{split}
\end{equation}
With the prior knowledge of an entangled pure state $\ket{\psi}$, there exists a generic construction of an EW operator to detect entanglement around $\ket{\psi}$. One widely used EW operator is given by
\begin{equation}\label{eq:projector_EW}
\mathcal{W}_{\psi} = \alpha \mathbb{I} - \ket{\psi}\bra{\psi},
\end{equation}
where $\alpha$ is the maximal fidelity between the state $\ket{\psi}$ and any bi-separable state \cite{MBourennane2004}. Here, $\mathcal{W}_{\psi}$ is a global EW operator and can detect GME of an unknown state $\rho$ close to $\ket{\psi}$. From the view of experimental implementation, two central concerns naturally arise. First, the EW operator is global and needs to be decomposed into local observables. Second, the generation of entanglement inevitably suffers from noises and hence the EW should be robust against the noises.
A few of quantum states with
certain symmetry including cluster states and more general graph states \cite{toth2005detecting,zhou2019detecting}, Greenberger-Horne-Zeilinger (GHZ)
states and W states \cite{ToolboxEWfor2007,Zhao2019Efficient} and more general permutation invariant states \cite{toth2009practical,zhou2019decomposition} are widely studied, which have efficient EW constructions. There are also a few other entanglement detection and quantification methods besides linear witness operators, such as entanglement criteria based on correlations \cite{Piotr2008, Li2017}, concurrences \cite{Mintert2007, Schmid2008}, Fisher information \cite{Hauke2016, Kourbolagh2019}, and spin squeezing from collective measurements \cite{sorensen2001many,sorensen2001entanglement,lucke2014detecting}.
Nevertheless, the entangled states generated by experimentally engineered Hamiltonians are usually complex and do not belong to standard Pauli stabilizer states. Simple EW construction methods cannot be directly applied here. How to efficiently witness GME for a relatively complex entangled state remains an open problem.

In this paper, we generalize the stabilizer formalism and apply the EW method to characterize the entanglement of our target states. We analyze the stabilizer structures and draw a family of EW operators which can be evaluated using a specific number of local measurements. Then, we design a searching algorithm to optimize the EW for the target state. In our method, optimization runs in the sense of finding the EW operator with the maximal noise tolerance. Here, the noise tolerance means the upper bound of tolerable noise for the experiment setup. Thus, given any noise condition, our method can always find an optimal EW operator with minimal local measurements to characterize GME for our target state.  In fact, our algorithm is general and can be used to design an optimal EW operator for other generalized stabilizer states. In particular, we find that the noise tolerance increases as the local measurement complexity of an EW operator grows. However, one cannot unboundedly increase the number of local measurement settings for the sake of the detection robustness. There exists a balance between detection efficiency and robustness. We investigate the trade-off between experiment efficiency and detection robustness to help find the most appropriate EW operator in the real laboratory.

The paper is organized as follows. Sec.~\ref{Sec:state_generation} introduces an efficient generation scheme of the large-scale entangled state. In Sec.~\ref{Sec:general_stabilizer}, we generalize the stabilizer formalism and then in Sec.~\ref{Sec:GME_witness} we apply it to construct the EW for the non-Pauli stabilizer state. Sec.~\ref{Sec:algorithm} proposes a generic searching algorithm to find the optimal EW operator. In Sec.~\ref{Sec:simulation}, we show the numerical results and discuss the trade-off between the detection robustness of a witness and its required number of measurement settings. Finally, we conclude and discuss possible future directions.

\section{Entanglement Generation}\label{Sec:state_generation}
Entangling gates are essential to the entanglement generation. The control of qubit interactions is the core technology of generating two-qubit or multi-qubit gates, where great efforts have been devoted to this direction \cite{jonathan2000fast, wang2001entanglement, yamamoto2010quantum}. Typical interactions arising in most experimental architectures include the Ising model \cite{Ising1925} and Heisenberg model \cite{Heisenberg1928}. The Ising interaction is usually implemented in trapped ions \cite{Kim_2011}, Rydberg atoms \cite{Labuhn2016}, nuclear magnetic resonance \cite{ryan2008liquid}, flux and charge superconducting systems \cite{orlando1999superconducting, makhlin1999josephson}, and can generate controlled-phase gate and controlled-NOT gate by tuning the time-evolution parameters of the spin-spin interaction. The Heisenberg interaction is widely employed in many experimental architectures and can generate various entangling gates. In this paper, we focus on the Heisenberg interaction, but our method can be applied to the Ising interaction as well.

The Heisenberg model is frequently employed to quantum dot spins \cite{loss1998quantum}, nuclear spins \cite{Simon2008}, and cavity QED \cite{imamog1999quantum}. The coupling Hamiltonian of the Heisenberg interaction is given by
\begin{equation}\label{eq:heisenberg}
H^{(ij)} = \frac{J}{2}(X_i X_j + Y_i Y_j + Z_iZ_j),
\end{equation}
where $J$ is the coupling strength between the qubits $i$ and $j$ and we denote $I_j, X_j, Y_j$ and $Z_j$ as the identity and Pauli operators for the qubit $j$. The time evolution of the Heisenberg interaction is given by
\begin{equation}
U^{(ij)}(t) = \exp[-iH^{(ij)}t] =
\left(
  \begin{matrix}
    e^{-iJt/2} & 0 & 0 & 0  \\
    0 & \cos(Jt)e^{iJt/2} & -i\sin(Jt)e^{iJt/2} & 0  \\
    0 & -i\sin(Jt)e^{iJt/2} & \cos(Jt)e^{iJt/2} & 0  \\
    0 & 0 & 0 & e^{-iJt/2}  \\
  \end{matrix}
\right)
\end{equation}
By tuning the evolution time at $t = \frac{\pi}{4J}$, a two-qubit $\sqrt{\mathrm{SWAP}}$ gate is obtained. In this paper, we will employ $\sqrt{\mathrm{SWAP}}$ gate to generate large-scale entanglement, and our method can be applied to other two-qubit entangling gates as well. One can use the $\sqrt{\mathrm{SWAP}}$ gates with several additional single qubit rotations to generate controlled-phase and controlled-NOT gates \cite{schuch2003natural}. Hence, the $\sqrt{\mathrm{SWAP}}$ gate, together with single-qubit gates, forms the complete gate set for universal quantum computing \cite{tanamoto2009efficient}. As mentioned above, the controlled-phase gate and controlled-NOT gate are unable to achieve directly in the Heisenberg interaction. The construction of the controlled gates pays a considerable cost in terms of the circuit depth and the number of qubit operations. Thus, one would better directly use $\sqrt{\mathrm{SWAP}}$ to generate entangled states, which is simpler, more efficient and robust to quantum noises. It is worth noting that the generated target state is not a standard stabilizer state, which is uniquely identified by a set of $n$-fold Pauli tensors.

The connection structure is similar to the Affleck-Kennedy-Lieb-Tasaki state \cite{PhysRevLett.59.799}, as shown in Fig.~\ref{Fig:AKLT}. Each pair of spin-1/2 is set to the quantum singlet, whereas the coupling operation is the $\sqrt{\mathrm{SWAP}}$ gate. The generation procedure runs as follows: First prepare the original register as $N$ singlets $\ket{\Psi_0} = \ket{\Psi^-}^{\otimes N}$, where $\ket{\Psi^-} = (\ket{01} - \ket{10})/ \sqrt{2}$, and then simultaneously apply the $\sqrt{\mathrm{SWAP}}$ gate between the nearest neighbor qubits.

\begin{figure}[hbtp!]
\begin{tikzpicture}[
    scale=1,
    interact/.style={blue,very thick},
    atom/.style={shade,shading=ball,circle,ball color=gray!10!},
    iswap/.style = {shape=ellipse,draw,dashed,color=red,thick,minimum width=50,minimum height=25,label=below:$\sqrt{SWAP}$},
]
  \node[atom] (obj1) at (-3.5,0) {};
  \node[atom] (obj2) at (-2.5,0) {};
  \node[atom] (obj3) at (-1.5,0) {};
  \node[atom] (obj4) at (-.5,0) {};
  \node[atom] (obj5) at (.5,0) {};
  \node[atom] (obj6) at (1.5,0) {};
  \node[atom] (obj7) at (2.5,0) {};
  \node[atom] (obj8) at (3.5,0) {};
  \draw[-,interact] (-4.5,0) -- (obj1);
  \draw[-,interact] (obj2) -- (obj3);
  \draw[-,interact] (obj4) -- (obj5);
  \draw[-,interact] (obj6) -- (obj7);
  \draw[-,interact] (obj8) -- (4.5,0);
  \node[iswap] at (-3,0) {};
  \node[iswap] at (-1,0) {};
  \node[iswap] at (1,0) {};
  \node[iswap] at (3,0) {};
  \node[atom] (obj9) at (-1.3,-1.5) {};
  \node[atom] (obj10) at (-0.3,-1.5) {};
  \draw[-,interact] (obj9) -- (obj10) node[right] {};
  \node[]  at (1.5,-1.5) {$=(\ket{01}-\ket{10})/\sqrt{2}$};
\end{tikzpicture}
\caption{(Color Online) Genuine multipartite entangled target state. The qubit pairs connected by solid lines denote singlets and the dashed ovals are $\sqrt{\mathrm{SWAP}}$ gates.}\label{Fig:AKLT}
\end{figure}
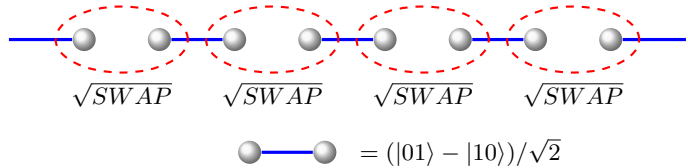

Finally, we achieve the $2N$-qubit entangled state as the target state,
\begin{equation}\label{eq:target_state}
\ket{\Psi} = \bigotimes_{i = 1}^N \sqrt{\mathrm{SWAP}}_{2i, 2i + 1}\ket{\Psi^-}^{\otimes N}.
\end{equation}
Note that the target state obtained here has the periodic boundary condition, i.e., the first and last qubits are coupled by the $\sqrt{\mathrm{SWAP}}$ gate. This may be unrealistic in some experiments when the qubits are positioned in a straight line and the interactions could only be employed between the nearest neighboring qubits. However, the lack of periodic boundary condition will not affect our analysis much. We will elaborate more on this in Sec.~\ref{Sec:GME_witness}. Furthermore, the target state itself has an interesting property. Since the Heisenberg model of Eq.~\eqref{eq:heisenberg} has a $SU(2)$ symmetry, the magnetic quantum number $\sum Z_i$ of the whole system is constant and maintains 0 during the entanglement generation procedures. Also, the whole generation procedures only contain a depth-1 circuit, which is resource-efficient for the experiment and can achieve very high fidelity.

\section{Generalized stabilizer formalism}\label{Sec:general_stabilizer}
Let $S_i$ denote the stabilizer operator on an $n$-qubit quantum system. A stabilizer $S_i$ is an $n$-fold tensor product of $n$ operators chosen from the one qubit Pauli operators $P_1 = \pm\{X, Y, Z, \mathbb{I}\}$. A stabilizer set $S = \{S_1, \cdots , S_n\}$ consisting of $n$ mutually commuting and independent stabilizer operators is called  the set of stabilizer ``generators''. The $n$ operators in set $S$ uniquely identify a state $\ket{\psi}$ satisfying $S_i\ket{\psi} = \ket{\psi}$ for $i = 1, \cdots, n$. Therefore, the density matrix of $\ket{\psi}$ can be written as $\ket{\psi}\bra{\psi} = \prod_i \frac{S_i + \mathbb{I}}{2}$ \cite{toth2005entanglement}. The entangled state we prepare in Sec.~\ref{Sec:state_generation} is not a standard stabilizer state, since the $\sqrt{\mathrm{SWAP}}$ operations transform the single Pauli operator to the summation of Pauli operators in $P_1^{\otimes n}$. This motivates us to generalize the definition of standard stabilizer state to the \emph{generalized stabilizer state} \cite{Nielsen2011Quantum}.

\begin{definition}[Generalized stabilizer state]\label{def:generalized_stabilizer}
For an $n$-qubit quantum system, a generalized stabilizer state $\ket{\Psi}$ is the unique eigenstate to eigenvalue $+1$ of the $n$ mutually commuting and independent generalized stabilizer operators $\{S_1, \cdots, S_n\}$, where each $S_i$ is an arbitrary Hermitian and unitary operator. Set $S$ is called the set of generalized stabilizer generators.
\end{definition}

\begin{Proposition}
For any pure state $\ket{\varphi} \in \mathcal{H}_2^{\otimes n}$, there exists a generalized stabilizer set $\{S_\varphi\}$ which can uniquely determine $\ket{\varphi}$.
\end{Proposition}

\begin{proof}
For $\ket{\varphi} \in \mathcal{H}_2^{\otimes n}$, we have $\ket{\varphi} = U_\varphi \ket{0}^{\otimes n}$, where $U_\varphi$ is a specific unitary operator determined by $\ket{\varphi}$. Since $S = \{Z_i\}$ is the stabilizer set of $\ket{0}^{\otimes n}$, we can derive $S_\varphi = U_\varphi \{Z_i\}U_\varphi^\dag$ as the stabilizer set of $\ket{\varphi}$.
\end{proof}

From Definition \ref{def:generalized_stabilizer}, we can easily construct a stabilizer set under any unitary transformation. Designate $\ket{\psi}$ as a standard stabilizer state and $S = \{S_i\}$ as its stabilizer generator set. Applying an arbitrary unitary $U$ on $\ket{\psi}$, we have $(US_iU^\dag) U\ket{\psi} = U\ket{\psi}$ for all $i$'s. In other words, set $S_U = U\{S_i\}U^\dag$ is the generator set for the generalized stabilizer state $U\ket{\psi}$, where $U\{S_i\}U^\dag$ represents premultiplying $U$ and postmultiplying $U^\dag$ to each element in $S$. There is a similar definition of generalized stabilizer state mentioned in Ref. \cite{plenio2007remarks} where the stabilizer $S_i$ is defined as an $N$-fold tensor product of arbitrary possibly non-Hermitian, linear operator. Our definition loosens the restriction of $N$-fold tensor product formulation, but limits the stabilizer $S_i$ to a Hermitian operator. We also remark that this generalized stabilizer formalism is used to construct the witness for the $W$ state \cite{guhne2009entanglement} and to quantify quantum coherence in the multipartite system \cite{Ding2020coherence}.

Now return to our target entangled state $\ket{\Psi}$ defined in Eq.~\eqref{eq:target_state}, which is constructed by $2N$ singlets with $\sqrt{\mathrm{SWAP}}$ operations connecting every neighboring pair of sites. Each singlet $\ket{\Psi^-}$ is a standard stabilizer state, whose stabilizer set can be chosen as $\{-XX, -ZZ\}$. Then consider the set of $\sqrt{\mathrm{SWAP}}$ operations as an additional operation $U$ applied on the $2N$ singlets, where $U = \bigotimes_{i = 1}^N \sqrt{\mathrm{SWAP}}_{2i, 2i + 1}$. Hence, the corresponding transformed stabilizer set for each singlet become $U\{-XX, -ZZ\}U^\dag$, which we define as
\begin{equation}\label{eq:U_transform_stabilizers}
\begin{split}
\tilde{S}^{(i)}_{XX} &= -UX_jX_{j+1}U^\dag \\
&= -(X_{j-1}\mathbb{I}_j+\mathbb{I}_{j-1}X_j - Y_{j-1}Z_j + Z_{j-1}Y_j) \otimes (X_{j+1}\mathbb{I}_{j+2}+\mathbb{I}_{j+1}X_{j+2} + Y_{j+1}Z_{j+2} - Z_{j+1}Y_{j+2}), \\
\tilde{S}^{(i)}_{ZZ} &= -UZ_jZ_{j+1}U^\dag \\
&= -(Z_{j-1}\mathbb{I}_j+\mathbb{I}_{j-1}Z_j + Y_{j-1}X_j - X_{j-1}Y_j) \otimes (Z_{j+1}\mathbb{I}_{j+2}+\mathbb{I}_{j+1}Z_{j+2} - Y_{j+1}X_{j+2} + X_{j+1}Y_{j+2}),
\end{split}
\end{equation}
with $j = 2i - 1$. Here $\tilde{S}^{(i)}_{XX}$ and $\tilde{S}^{(i)}_{ZZ}$ are the $U$-transformed stabilizers for the $i$-th singlet. In general, the single Pauli product term $X_jX_{j+1} (Z_jZ_{j+1})$ are transformed into the linear combinations of 16 Pauli tensors under the entangling operation $U$.

\section{GME witness for the generalized stabilizer state}\label{Sec:GME_witness}
In order to study the entanglement properties of a given state, one may conventionally employ the tomography method to fully characterize the quantum state. As we known, quantum state tomography is resource-intensive and becomes impractical for large-scale quantum systems. The EW method is proposed to balance the measurement complexity and the completeness of state information, which can efficiently verify the entanglement for any given state \cite{guhne2009entanglement}. Before giving detailed descriptions of the EW method, we first give a rigorous definition of GME. A pure state $\ket{\psi}$ is called bi-separable if we can find a bi-partition $\{A, B\}$ such that $\ket{\psi} = \ket{\phi}_A \otimes \ket{\chi}_B$. A mixed state $\rho$ is called bi-separable if it can be decomposed into a convex combination of the pure product states $\ket{\psi_i} = \ket{\phi}_{A_i} \otimes \ket{\chi}_{B_i}$ under bipartitions $\{A_i, B_i\}$, i.e., $\rho = \sum_i p_i \ket{\psi_i}\bra{\psi_i}$ and each $\ket{\psi_i}$ can have different bipartition. Otherwise, the state possesses genuine multipartite entanglement \cite{Horodecki2009entanglement,guhne2009entanglement}.

\subsection{Projector-based witness}
As for our target state $\ket{\Psi}$ of Eq.~\eqref{eq:target_state}, the projector-based EW operator can be written as
\begin{equation}\label{eq:original_EW}
\mathcal{W}_{\Psi} = \frac{5}{8} \mathbb{I} - \ket{\Psi}\bra{\Psi},
\end{equation}
Here, the coefficient $\alpha_{\Psi} = 5/8$ is computed as the largest Schmidt coefficient of $\ket{\Psi}\bra{\Psi}$ under any bipartition \cite{MBourennane2004}. Note that for our target state, this value remains constant and has nothing to do with the amount $N$ of singlets. The operator $\mathcal{W}_{\psi}$ is a global witness and can detect GME of an unknown state $\rho$ close to $\ket{\Psi}$. One can always conclude that $\rho$ is genuinely multipartite entangled when its fidelity satisfies $F_{\rho} = \tr(\ket{\Psi}\bra{\Psi}\rho) > 5/8$.

Then the question comes to how to measure the witness operator $\mathcal{W}_{\Psi}$. Now we show that the projector-based witness in Eq.~\eqref{eq:original_EW} can be decomposed as a summation of the generalized stabilizers terms. As mentioned in Sec.~\ref{Sec:general_stabilizer}, target state $\ket{\Psi}$ is the unique $+1$ eigenstate of each element in the stabilizer set $\{\tilde{S}^{(i)}_{XX}, \tilde{S}^{(i)}_{ZZ}\}$, thus, we can define the corresponding stabilizer projectors,
\begin{equation}\label{eq:stabilizer_projector}
\begin{split}
\tilde{P}^{(i)}_{XX} = \frac{\tilde{S}^{(i)}_{XX} + \mathbb{I}}{2}, \\
\tilde{P}^{(i)}_{ZZ} = \frac{\tilde{S}^{(i)}_{ZZ} + \mathbb{I}}{2}.
\end{split}
\end{equation}
Then, the projectors onto the target state can be written as the product of the stabilizing projectors,
\begin{equation}
\ket{\Psi}\bra{\Psi} = \prod_i \tilde{P}^{(i)}_{XX} \tilde{P}^{(i)}_{ZZ}.
\end{equation}
From the above definition, we can see the target state projector can always be decomposed into a combination of products of stabilizer terms. One can measure these observables in experiments and obtain an estimation value of $\mathcal{W}_{\Psi}$.

The key challenge of the EW method is to efficiently estimate the expectation value of $\mathcal{W}_{\Psi}$ in experiments. The experimental efforts for measuring a witness can be described by the number of local measurements. Denote $\bigotimes_{k = 1}^n \mathcal{O}^{(k)}$ as a local measurement setting (LMS), which consists of performing Pauli measurements $\mathcal{O}^{(k)} \in \{X, Y, Z\}$ or doing nothing ($\mathcal{O}^{(k)} = \mathbb{I}$) to the $k$-th qubit \cite{Guhne2002}. Note that the measurements in one LMS are performed simultaneously. For an EW operator $\mathcal{W}$, we call the number of LMSs required to measure $\mathcal{W}$ as its \emph{local measurement complexity} (LMC) \cite{zhou2019decomposition}, denoted by $C_\mathcal{W}$, which is an important quantity to evaluate the efficiency of $\mathcal{W}$ in experiments. For typical symmetric entangled states, GHZ and $W$ states, one can make efficient decompositions of the corresponding state projectors and estimate the fidelities with $N + 1$ and $2N + 1$ LMSs, respectively. However, for some more complex entangled states, the fidelity estimation requires too much measurement effort even using state decomposition. In this case, instead of measuring the projector-based witness $\mathcal{W}_\psi$ of Eq.~\eqref{eq:projector_EW}, one can alternatively measure a new witness \cite{ToolboxEWfor2007}
\begin{equation}
\mathcal{W}'_\psi = \alpha \mathbb{I} - \ket{\psi}\bra{\psi} + Q,
\end{equation}
where $Q$ is a positive operator, and the criterion in Eq.~\eqref{eq:EW_criterion} can always be satisfied. Then the question comes to how to choose the positive operator $Q$ so that the estimation of witness requires only a few measurements, whereas still obtaining a good bound on the fidelity. One feasible solution is to subtract the complicated state projector $\ket{\psi}\bra{\psi}$ and replace it by the linear combination of the operators with fewer LMSs $\{P_i\}$, i.e.,
\begin{equation}
Q = \ket{\psi}\bra{\psi} + \sum_ic_i{P_i}.
\end{equation}
Note that the operator set $\{P_i\}$ is indeed a decomposition of the projector, i.e., $\ket{\psi}\bra{\psi} = \prod_i P_i$, thus we call it the \emph{bounding decomposition method}.

Now return to our target state. Unlike standard Pauli stabilizer states, the $\sqrt{\mathrm{SWAP}}$ operations make the stabilizers of our target state more complicated [see Eq.~\eqref{eq:U_transform_stabilizers}], which may increase the measurement effort for estimating the fidelity. It is easy to see the LMC for the projector-based EW operator of Eq.~\eqref{eq:original_EW} is $C_{\mathcal{W}_\psi} = 3^{2(N - 1)}$, comparable to the experimental cost for quantum state tomography. Obviously, this witness is impractical for a large-scale quantum system. One can derive an appropriate operator $Q$ via the bounding decomposition method to obtain a more efficient witness. We will demonstrate this with several cases in the next section.

\subsection{Witness operator with $3^k$ LMSs}
In this section, we will focus on the bounding decomposition method for the EW operator. We first review a simple decomposition scheme for projectors \cite{toth2005detecting, zhou2019detecting}.
\begin{Proposition}\label{prop:P1P2}
For a set of projectors $\{P_1, \cdots, P_k\}$, we have
\begin{equation}\label{eq:P1P2}
 P_1 \cdots P_k \geq  P_1 + \cdots + P_k - (k - 1)\mathbb{I},
\end{equation}
where $A \geq B$ means $(A - B)$ is positive semidefinite.
\end{Proposition}
In the following part we will omit the brackets when there is no ambiguity. As for the inequality of Eq.~\eqref{eq:P1P2}, the LMC needed to measure the left part is larger than that for the right part. For example, set $P_1 = (\mathbb{I} + X)/2, P_2 = (\mathbb{I} + Z)/2$, the LMC for measuring $P_1P_2$ is $C_L = 3$, whereas for measuring $(P_1 + P_2)$ is $C_R = 2$. Thus, one can improve the efficiency of an EW operator by changing the product form into the summation form. Since no overlap exists between the LMSs for the transformed stabilizers $\tilde{S}^{(i)}_{XX}$ and $\tilde{S}^{(i)}_{ZZ}$ of Eq.~\eqref{eq:U_transform_stabilizers}, one should separate the corresponding projectors to decrease the total LMC for the witness. According to Proposition \ref{prop:P1P2}, one can choose $Q = \ket{\Psi}\bra{\Psi} + 2\mathbb{I} - \prod_i^N \tilde{P}^{(i)}_{XX} - \prod_i^N \tilde{P}^{(i)}_{ZZ}$ and propose a more efficient witness
\begin{equation}\label{eq:XZ_EW}
\mathcal{W}'_{\Psi} = \frac{13}{8} \mathbb{I} - (\prod_i^N \tilde{P}^{(i)}_{XX} + \prod_i^N \tilde{P}^{(i)}_{ZZ}),
\end{equation}
with the LMC $C_{\mathcal{W}_{\Psi}'} = 2\cdot 3^{(N - 1)}$, whereas the LMC for the original EW operator defined in Eq.~\eqref{eq:original_EW} is $C_{\mathcal{W}_{\Psi}
} = 3^{2(N - 1)}$. Our slight modification makes a square root reduction for the LMC. Moreover, the product terms $\prod_i^N \tilde{P}^{(i)}_{XX}$ and $\prod_i^N \tilde{P}^{(i)}_{ZZ}$ can be further decomposed based on Proposition \ref{prop:P1P2}. In general, we can present a projector product subset with different LMSs for the EW operator, as shown in the following proposition.

\begin{Proposition}\label{prop:projector_subset}
For the target state of Eq.~\eqref{eq:target_state} and its corresponding stabilizer projectors $\{\tilde{P}^{(i)}\}_{XX/ZZ}$ defined in Eq.~\eqref{eq:U_transform_stabilizers} and \eqref{eq:stabilizer_projector}, the LMC of measuring the projector product term $\tilde{P}^{(i_1)} \cdots \tilde{P}^{(i_k)} (1 \leq i_1 < \cdots < i_k \leq N)$ is given by,
\begin{equation}\label{eq:noLMC}
\begin{aligned}
LMC = \left\{
             \begin{array}{ll}
               9, & k=1; \\
               15, & k=2 \;\mathrm{and}\; (i_1-i_2) \;\mathrm{mod}\; N = \pm1; \\
               3^{\sum_{r = 1}^k \min(i_r - i_{r - 1}, 2)}, & \hbox{otherwise,}
             \end{array}
           \right.
\end{aligned}
\end{equation}
where $i_0$ is set to be $i_k-N$ with the periodic boundary condition.
\end{Proposition}

\begin{proof}
We only prove the case of $\{\tilde{P}^{(i)}\}_{XX}$ and the same argument applies for $\{\tilde{P}^{(i)}\}_{ZZ}$. We omit the subscript $XX$ for simplicity. We first consider a single projector $\tilde{P}^{(i)} = (\tilde{S}^{(i)} + \mathbb{I}) / 2$, whose LMC is explicitly determined by $\tilde{S}^{(i)}$. As shown in Eq.~\eqref{eq:U_transform_stabilizers}, there are 16 Pauli combinations in $\tilde{P}^{(i)}$. Let $j=2i-1$, a good point is that only one LMS $X^{(j - 1)}X^{(j)}$ is needed for measuring $X_{j-1}I_j$ and $I_{j-1}X_j$. Thus we only need nine LMSs to estimate the projector $\tilde{P}^{(i)}$, denoted by $\{O^{(i)}_1 \otimes O^{(i)}_2\}$, where
\begin{equation}\label{eq:O1O2}
\begin{split}
&O^{(i)}_1 \in \{X^{(j - 1)}X^{(j)}, Y^{(j - 1)}Z^{(j)}, Z^{(j - 1)}Y^{(j)}\}, \\
&O^{(i)}_2 \in \{X^{(j + 1)}X^{(j + 2)}, Y^{(j + 1)}Z^{(j + 2)}, Z^{(j + 1)}Y^{(j + 2)}\}.
\end{split}
\end{equation}
Then, consider a projector product term $\tilde{P}^{(i)}\tilde{P}^{(i + 1)}$, the corresponding LMC is determined by $\tilde{S}^{(i)}, \tilde{S}^{(i + 1)}$ and $\tilde{S}^{(i)}\tilde{S}^{(i + 1)}$. Consistent with Eq.~\eqref{eq:O1O2}, the LMSs for measuring $\tilde{S}^{(i)}, \tilde{S}^{(i + 1)}$ can be written as $\{O^{(i)}_1 \otimes O^{(i)}_2\}$ and $\{O^{(i+1)}_1 \otimes O^{(i+1)}_2\}$. Due to the identical relation,
\begin{equation}
U X^{(j + 1)}X^{(j+2)} U^\dag = X^{(j + 1)}X^{(j+2)},
\end{equation}
where $U$ is the coupling operation discussed in Sec.~\ref{Sec:general_stabilizer}, the LMSs for measuring $\tilde{S}^{(i)}\tilde{S}^{(i + 1)}$ are $\{O^{(i)}_1 \otimes X^{(j+1)}X^{(j+2)}\otimes \mathcal{O}^{(i+1)}_2\}$. Thus the total LMC for $\tilde{P}^{(i)}\tilde{P}^{(i + 1)}$ is 15.
\\
Finally we consider the projector product term $\tilde{P} = \tilde{P}^{(i_1)} \cdots \tilde{P}^{(i_k)}$ [$k \geq 2$, $(i_1-i_2) \;\mathrm{mod} \; N \neq \pm1$]. Decomposing $\tilde{P}$ into the sum of stabilizers, one can find the related LMSs are determined by the stabilizer set,
\begin{equation}
\mathcal{S} = \{s(a_1 \cdots a_k) | s(a_1 \cdots a_k) = (\tilde{S}^{(i_1)})^{a_1} \cdots (\tilde{S}^{(i_k)})^{a_k}, a_j \in \{0, 1\}, 1 \leq j \leq k\}.
\end{equation}
We now prove the LMSs $\mathcal{L} = \{O^{(i_1)}_1 \otimes O^{(i_1)}_2 \otimes \cdots \otimes O^{(i_k)}_1 \otimes O^{(i_k)}_2\}$ are necessary and sufficient for measuring $\tilde{P}$. The LMSs for measuring each stabilizer $\tilde{S}^{(i_j)}$ are $\{O^{(i_j)}_1 \otimes O^{(i_j)}_2\}$, thus, $\mathcal{L}$ are sufficient for all the elements in $\mathcal{S}$, i.e., LMS$(S) \subseteq \mathcal{L}$. Then, we prove the necessity of $\mathcal{L}$. One can always construct a subset $\mathcal{S}' \subseteq \mathcal{S}$, where the LMSs of $\mathcal{S}'$ cover all the elements of $\mathcal{L}$. The construction runs as follows. First set $a_1 = 1$ and track the order $j$ from 2 to $k$: If $a_{j-1}=1$ and $i_j = i_{j - 1} + 1$, set $a_j = 0$; otherwise, set $a_j = 1$. Then one obtain a specific stabilizer product $s(a_1 \cdots a_k)$. Second collect the single stabilizer terms $\tilde{S}^{(i_j)}$ with $a_j = 0$. The subset is constructed as
\begin{equation}
\mathcal{S}' = \{s(a_1 \cdots a_k)\} \cup \{\tilde{S}^{(i_j)} | a_j = 0 \} \cup \{\tilde{S}^{(i_{j-1})} \tilde{S}^{(i_j)} | a_j = 0 \}.
\end{equation}
It is not difficult to see that LMS$(\mathcal{S}')$ is equivalent to  $\mathcal{L}$. Since $\mathcal{S}' \subseteq \mathcal{S}$, one can conclude that $\mathcal{L} \subseteq$ LMS$(\mathcal{S})$. Thus, there completes the proof of  $\mathcal{L} = LMS(\mathcal{S})$. The LMC for $\tilde{P}$ is determined by $|\mathcal{L}| = 3^{\sum_{r = 1}^k \min(i_r - i_{r - 1}, 2)}$.
\end{proof}

Here we formalize Proposition \ref{prop:projector_subset} with the periodic boundary condition,  i.e., the first and last qubit can be coupled by the entangling gate. One can directly generalize the result here to open boundary condition by taking $i_0 = i_1 - 2$ in Proposition \ref{prop:projector_subset}.

Clearly from Proposition \ref{prop:projector_subset}, we can obtain the EW operator with the minimal LMC by replacing the complicated terms $\prod_i^N \tilde{P}^{(i)}_{XX}$ and $\prod_i^N \tilde{P}^{(i)}_{ZZ}$ by a sum of single projector terms, i.e., $\mathcal{W}^{''}_\Psi = (2N - \frac{3}{8})\mathbb{I} - \sum_i^N (\tilde{P}^{(i)}_{XX} + \tilde{P}^{(i)}_{ZZ})$. Thanks to the translation invariant symmetric structure of the target state, still only 9 LMSs are required to measure the summation term $\sum_i^N \tilde{P}^{(i)}_{XX}$. One can periodically select one of the three Pauli tensors for one LMS, i.e. $O^{(1)}_1 = \cdots = O^{(N)}_1, O^{(1)}_2 = \cdots = O^{(N)}_2$. The same is true for the projector $\tilde{P}^{(i)}_{ZZ}$. Hence, one needs 18 LMSs in total to measure the EW operator $\mathcal{W}^{''}_\Psi$.

It seems we obtain the most efficient EW operator present above. However, in real experiments, the quantum noises exist and may affect the detection performances. The more decompositions we perform, the more sensitive the EW operators will be to quantum noises. Actually, there exists a trade-off between the detection robustness and the experimental efficiency \cite{Zhao2019Efficient}. In the following section, we will design a searching algorithm for constructing the optimal EW operator with the maximal detection robustness. Afterwards, we can always find an appropriate entanglement detection strategy under any experimental condition.

In practice, the generation of the target state is suffering from non-negligible noises, which can affect the performance of entanglement detection. For simplicity, we treat the prepared noisy state as the target state mixed with the white noise
\begin{equation}\label{eq:noisy_target_state}
\rho_\Psi = (1 - p)\ket{\Psi}\bra{\Psi} + \frac{p}{2^{2N}}\mathbb{I}.\end{equation}
A valid EW operator for $\rho_\Psi$ must obey the criterion in Eq.~\eqref{eq:EW_criterion}, thus, there always exists an upper bound for $p$, denoted as $p_{\mathrm{max}}$, indicating the maximum tolerable noise error rate. We can treat $p_{\mathrm{max}}$ as an essential figure of merit for the detection robustness. Applying the EW operator $\mathcal{W}'_\Psi$ defined in Eq.~\eqref{eq:XZ_EW}, one can calculate,
\begin{equation}\label{eq:maxtolerablep}
p_{\mathrm{max}} = \frac{3}{16(1 - 2^{-N})},
\end{equation}
which approaches $3 / 16$ when $N \rightarrow \infty$. The product projectors $\prod_i^N \tilde{P}^{(i)}_{XX}$ and $\prod_i^N \tilde{P}^{(i)}_{ZZ}$ in $\mathcal{W}'_\Psi$ can be further decomposed to decrease the total LMC for GME detection. Meanwhile, the detection robustness $p_{\mathrm{max}}$ is changed as the EW operator changes.

To be specific, we will give an example to show how the formulations of EW operators affect $p_{\mathrm{max}}$. Recall that no overlap exists between the LMSs for $\{\tilde{P}^{(i)}_{XX}\}$ and $\{\tilde{P}^{(i)}_{ZZ}\}$. The decomposition of $\prod_i^N \tilde{P}^{(i)}_{XX}$ works for $\prod_i^N \tilde{P}^{(i)}_{ZZ}$ as well, thus, we will take the same decomposition method for these two product projectors. Set $N = 5$ without periodic boundary condition, a ten-qubit target state is obtained. Now fix the LMC for measuring the product projector $\tilde{P}_\gamma = \tilde{P}^{(1)}_\gamma\cdots \tilde{P}^{(5)}_\gamma$ to $3^4, \gamma \in \{XX, ZZ\}$, a direct application of Proposition \ref{prop:P1P2} is:
\begin{equation}\label{eq:ineq_1}
\tilde{P}_\gamma \geq \tilde{P}^{(1)}_\gamma\tilde{P}^{(2)}_\gamma\tilde{P}^{(3)}_\gamma + \tilde{P}^{(4)}_\gamma\tilde{P}^{(5)}_\gamma - \mathbb{I},
\end{equation}
then we obtain a new EW operator
\begin{equation}
\mathcal{W}^1_\Psi = \frac{29}{8}\mathbb{I} - \sum\limits_{\gamma\in\{XX, ZZ\}} \left( \tilde{P}^{(1)}_\gamma\tilde{P}^{(2)}_\gamma\tilde{P}^{(3)}_\gamma + \tilde{P}^{(4)}_\gamma\tilde{P}^{(5)}_\gamma \right),
\end{equation}
and its corresponding  maximum tolerable noise error rate is $p_{\mathrm{max}} = 11.5\%$. Another nontrivial inequality of the projectors is $(\mathbb{I} - \tilde{P}^{(1)}_\gamma\tilde{P}^{(2)}_\gamma)\tilde{P}^{(3)}_\gamma(\mathbb{I} - \tilde{P}^{(4)}_\gamma\tilde{P}^{(5)}_\gamma) \geq 0$, i.e.,
\begin{equation}\label{eq:ineq_2}
\tilde{P}_\gamma \geq \tilde{P}^{(1)}_\gamma\tilde{P}^{(2)}_\gamma\tilde{P}^{(3)}_\gamma + \tilde{P}^{(3)}_\gamma\tilde{P}^{(4)}_\gamma\tilde{P}^{(5)}_\gamma - \tilde{P}^{(3)}_\gamma.
\end{equation}
Hence we derive another EW operator
\begin{equation}\label{eq:EW2}
\mathcal{W}^2_\Psi = \frac{13}{8}\mathbb{I} - \sum\limits_{\gamma\in\{XX, ZZ\}} \left( \tilde{P}^{(1)}_\gamma\tilde{P}^{(2)}_\gamma\tilde{P}^{(3)}_\gamma + \tilde{P}^{(3)}_\gamma\tilde{P}^{(4)}_\gamma\tilde{P}^{(5)}_\gamma - \tilde{P}^{(3)}_\gamma \right),
\end{equation}
the  maximum tolerable noise error rate for $\mathcal{W}^2_\Psi$ is $p_{\mathrm{max}} = 15.0\%$. Although $\mathcal{W}^1_\Psi$ and $\mathcal{W}^2_\Psi$ have the same LMC, $\mathcal{W}^2_\Psi$ is more robust to quantum noises in real experiments. Actually, for a given value of LMC, there are a variety of the EW operators. One should choose one with optimal detection robustness.

\section{Algorithm for the constructing optimal witness}\label{Sec:algorithm}
In this section, we propose and apply a searching algorithm to construct the optimal witness under different noise levels. As shown in the previous example, optimizing the noise tolerance can be reduced to constructing the optimal positive operator $Q$ in the bounding decomposition method. In the following, we will describe our algorithm in a general version.

Given a set of projectors $\{P_1, \cdots, P_n\}$, assume the bounding decomposition inequality can be written as
\begin{equation}\label{eq:decompo_ineq}
P_1 P_2 \cdots P_n \geq P_d \equiv \sum\limits_{j}c_jP[a_j],
\end{equation}
where $a_j \in \mathbb{Z}_2^n, P[a_j] = P_1^{a_j(1)}P_2^{a_j(2)} \cdots P_n^{a_j(n)}$; the tuple $(c_j, P[a_j])$ consisting of the coefficient and product of projectors is the solution to be solved in the algorithm. Take the inequality in Eq.~\eqref{eq:ineq_1} for example, the solution tuples are $(1, \tilde{P}^{(1)}\tilde{P}^{(2)}\tilde{P}^{(3)}), (1, \tilde{P}^{(4)}\tilde{P}^{(5)}), (-1, \mathbb{I})$.

Based on the inequality of Eq.~\eqref{eq:decompo_ineq}, we can choose the positive operator $Q$ as
\begin{equation}
Q = \ket{\psi}\bra{\psi} - \sum\limits_{j}c_jP[a_j],
\end{equation}
here each $P[a_j]$ is constructed from the stabilizer projectors of Eq.~\eqref{eq:U_transform_stabilizers}. Note that a valid EW operator must obey the criterion in Eq.~\eqref{eq:EW_criterion}. Consider the existence of noise in real experiments [see Eq.~\eqref{eq:noisy_target_state}], we have
\begin{equation}\label{eq22}
\begin{split}
&\tr(P_d \rho_\Psi) = \frac{p}{2^{2N}}\sum\limits_{j \geq 1} c_j2^{2N-s[a_j]} + (1-p)\sum\limits_{j \geq 1}c_j + c_0 \geq \alpha, \\
&\mathrm{i.e.,\ } p \sum\limits_{j \geq 1}c_j(1-2^{-s[a_j]})\leq \sum\limits_{j \geq 0}c_j - \alpha \leq 1 - \alpha,
\end{split}
\end{equation}
where $c_0$ is related to $P[a_0] = \mathbb{I}$, and $s[a_j] = \sum_i^n a_j(i)$ is the Hamming weight of the $a_j$ vector. 
We derive the second inequality using the fact that $\sum_{j\geq 0} c_j \leq 1$, due to the positivity of $Q$. Explicitly, the detection robustness can be optimized by
\begin{equation}\label{eq:pmax}
p_{\mathrm{max}} = \frac{1 - \alpha}{\min \sum\limits_{j \geq 1}c_j(1-2^{-s[a_j]})}.
\end{equation}
Here the minimization is over all possible witness operators. For our target state, the detailed construction of the witness subset is shown in Proposition \ref{prop:projector_subset}. It is not hard to see that $2^{-s[a_j]}$ decreases quickly as $s[a_j]$ grows, indicating that the value of $\sum_{j \geq 1}c_j$ makes the main contribution to $p_{\mathrm{max}}$. As a result, we  divide the minimize procedure further into two components: (a) Minimize $\sum_{j \geq 1}c_j$; (b) minimize $s[a_j]$ for each $j$.

We begin by describing the minimization in (a). According to the fact that $\sum_{j \geq 1}c_j \leq 1 - c_0$, the problem of minimizing $\sum_{j \geq 1}c_j$ can be reduced to maximizing $c_0$. The value of $c_0$ denotes the weight of $\mathbb{I}$ appearing in the bounding decomposition inequality of Eq.~\eqref{eq:decompo_ineq}. Intuitively, this quantity can be seen as the number of decompositions. Thus we should avoid unnecessary divisions for the original state projector. For a given set of witness operators $\{P[a_j]\}$, we can extract the greatest common divisor projector to ensure the efficient decomposition. Assume the set $\{P[a_j]\}_{j \geq 1}$ in Eq.~\eqref{eq:decompo_ineq} has a common divisor projector $P_{cd}$, the bounding decomposition inequality can be written as
\begin{equation}\label{eq:decompo_ineq2}
\begin{split}
P_1 P_2 \cdots P_n  &\equiv P_{cd}\bar{P}_{cd} \\
&\geq P_{cd} \left(\sum_{j \geq 1} c'_j P[a'_j] + c'_0 \mathbb{I}\right).
\end{split}
\end{equation}
where $\bar{P}_{cd}$ is the complementary projector of $P_{cd}$, and $\bar{P}_{cd}$ itself can be further bounded by $\bar{P}_{cd} \geq \sum_{j \geq 1} c'_j P[a'_j] + c'_0 \mathbb{I}$. Then the inequality can be improved by
\begin{equation}\label{eq25}
P_1 P_2 \cdots P_n \geq \sum_{j \geq 1} c'_j P[a_j] + c'_0 P_{cd},
\end{equation}
where $P[a_j] = {P}_{cd}P[a'_j]$, and the original $c_0 = 0$ in Eq.~\eqref{eq25}. Since $c_0 \leq 0$, $0$ is the maximum value of $c_0$ if $P_{cd} \neq \mathbb{I}$. Note that one can repeat the above procedure for the complementary projector and solve the greatest common divisor in each iteration until $P_{cd} = \mathbb{I}$. After that, one can use Procedure (b) to minimize $\sum_{j \geq 1}c_j[1-2^{-s[a_j]}]$, as described below.

We now consider the minimization in (b). In order to minimize $s[a_j]$ for each $j$, i.e., reduce the number of product terms, we introduce a numerical truncation method. Recall that $P[a_j] = P_1^{a_j(1)}P_2^{a_j(2)} \cdots P_n^{a_j(n)}$ and denote $\mathrm{Supp}(a_j)$ as the Hilbert space involved by $P[a_j]$. One can remove a divisor projector $P[b_j]$ from $P[a_j]$ if the following two requirements are satisfied after removing it:
\begin{enumerate}
\item The common divisor projector $P_{cd}$ solved in Procedure (a) remains unchanged;

\item $\bigcup_{j \geq 1} \mathrm{Supp}(P[a'_j]) = \mathrm{Supp}(P_1P_2 \cdots P_n)$, where $P[a'_j] = P[a_j] / P[b_j]$.
\end{enumerate}
One can traverse all the divisor projectors in $P[a_j]$ and make the truncation procedure for all $j$'s. Then, a minimum $s[a_j]$ will be obtained.

\begin{observation}
A tight bounding decomposition inequality can be written as
\begin{align}
    \prod\limits_{P_i\in\mathcal{I}_1}P_i \cdot \prod\limits_{P_j\in\mathcal{I}_0}(\mathbb{I}-P_j) \geq 0,
\end{align}
where $\mathcal{I}_0\sqcup\mathcal{I}_1=\text{Supp}(P_1P_2 \cdots P_n)$ denotes the Hilbert space spanned by $P_1P_2 \cdots P_n$. \end{observation}

To summarize, for a target state projector $\ket{\psi}\bra{\psi}$ and its related witness subset $\{P[a_j]\}$, one can first apply Procedure (a) to find out the greatest common divisor projector, and then apply Procedure (b) to minimize $s[a_j]$ for all $j$'s. Finally one will obtain an EW operator with optimal detection robustness. The detailed algorithm for the optimization procedures is shown in Algorithm 1.

\begin{centering}
\begin{algorithm}[H]\label{algo1}
\SetAlgoLined
\LinesNumbered
\KwInput{An EW operator: $\mathcal{W} = \alpha \mathbb{I} - \ket{\psi}\bra{\psi}$, value of LMC: $l$.}
\KwOutput{An EW operator with optimal detect robustness and at most $l$ LMC: $\mathcal{W}_{\mathrm{opt}}$.}
Construct a set of orthonormal basis $\{\ket{e_k}\}_0^{2^n-1}$ of $\mathcal{H}_n$ with $\ket{e_0}=\ket{\psi}$\;
Denote $U=\sum\limits_{k=0}^{2^n-1}\ket{e_k}\bra{k}$ and set $P_i=UZ_iU^\dagger$, $\forall i\in[n]$\;
Construct $\mathcal{T}^* = \{q_j\}$, where $q_j$ is the product of elements in $\{P_i\}_1^n$ with $\mathrm{LMC}(q_j) \leq l, \forall j$ and $\text{Supp}(\mathcal{T}^*)=\text{Supp}(\prod\limits_{i=1}^nP_i)$\;

$P_{cd}\leftarrow\mathbb{I}$\;

\For{$i$ \text{from} 1 \text{to} $n$}{
\If{$P_i$ is a divisor of all the elements in $\mathcal{T}^*$}
{
    $P_{cd}\leftarrow P_{cd}\cdot P_i$;}
}

Truncate terms in $\mathcal{T^*}$\;

\While{no further truncation could be down}{
   Starting from the term with the highest degree in $\mathcal{T^*}$, remove redundant sub-projectors;
 }

Return $\mathcal{W}_{\mathrm{opt}} = \alpha\mathbb{I} - \sum\limits_{q_j\in\mathcal{T^*}} q_j + (|\mathcal{T^*}|-1) P_{cd}$.\label{line14}
\caption{Searching algorithm for constructing optimal EW operators.}
\end{algorithm}
\end{centering}


We show a concrete example to explain how our algorithm works. For a ten-qubit target state, we make the same bounding decompositions for $\prod_i^N \tilde{P}^{(i)}_{XX}$ and $\prod_i^N \tilde{P}^{(i)}_{ZZ}$. As for the product projector $\prod_i^N \tilde{P}^{(i)}_{\gamma}, \gamma \in \{XX, ZZ\}$, the witness set with LMC $l \leq 27$ is
\begin{equation}
\begin{split}
\mathcal{T}_\gamma = \{\tilde{P}^{(1)}_\gamma\tilde{P}^{(2)}_\gamma\tilde{P}^{(3)}_\gamma, \tilde{P}^{(3)}_\gamma \tilde{P}^{(4)}_\gamma \tilde{P}^{(5)}_\gamma, \tilde{P}^{(1)}_\gamma \tilde{P}^{(2)}_\gamma, \tilde{P}^{(2)}_\gamma \tilde{P}^{(3)}_\gamma, \\ \tilde{P}^{(3)}_\gamma \tilde{P}^{(4)}_\gamma, \tilde{P}^{(4)}_\gamma \tilde{P}^{(5)}_\gamma, \tilde{P}^{(1)}_\gamma \tilde{P}^{(5)}_\gamma , \tilde{P}^{(1)}_\gamma ,\tilde{P}^{(2)}_\gamma ,\tilde{P}^{(3)}_\gamma ,\tilde{P}^{(4)}_\gamma ,\tilde{P}^{(5)}_\gamma\}.
\end{split}
\end{equation}
Make a trial and choose
\begin{align}
    \mathcal{T^*}=\{\tilde{P}^{(1)}_\gamma\tilde{P}^{(2)}_\gamma\tilde{P}^{(3)}_\gamma, \tilde{P}^{(2)}_\gamma\tilde{P}^{(3)}_\gamma, \tilde{P}^{(3)}_\gamma\tilde{P}^{(4)}_\gamma, \tilde{P}^{(3)}_\gamma\tilde{P}^{(4)}_\gamma\tilde{P}^{(5)}_\gamma\},
\end{align}
where the greatest common divisor is $P_{\mathrm{cd}} = \tilde{P}^{(3)}_\gamma$. We can remove $\tilde{P}^{(2)}_\gamma\tilde{P}^{(3)}_\gamma$ and $\tilde{P}^{(3)}_\gamma\tilde{P}^{(4)}_\gamma$ to decrease $|\mathcal{T^*}|$ by 2 whereas keeping $P_{\mathrm{cd}} = \tilde{P}^{(3)}_\gamma$ unchanged. No more truncation can be performed, and we finally obtain $\mathcal{T^*}=\{\tilde{P}^{(1)}_\gamma\tilde{P}^{(2)}_\gamma\tilde{P}^{(3)}_\gamma, \tilde{P}^{(3)}_\gamma\tilde{P}^{(4)}_\gamma\tilde{P}^{(5)}_\gamma\}$.  Thus we can construct the same inequality as Eq.~\eqref{eq:ineq_2} and compute the optimal EW operator of Eq.~\eqref{eq:EW2}.

\section{Numerical simulation}\label{Sec:simulation}
In this section, we apply our searching algorithm to the target state and construct optimal EW operators under different noise levels.

Recall that the LMSs for each stabilizers are formulated by $O_1 \otimes O_2$, where $O_1, O_2 \in \{XX, YZ, ZY\}$. Different product combinations of projectors require different periodic selections of Pauli tensors in $O_1$ and $O_2$, which decide the corresponding LMSs. Since the LMSs are determined when we fix the LMC, one thing we should do is to find the optimal EW operator with the maximal noise tolerance $p_{\mathrm{max}}$. In the following, we use $F_{\mathrm{min}} = 1 - p_{\mathrm{max}}$ as the figure of merit, which is nearly the minimum requirement for the experimental fidelity. As for the target state defined in Eq.~\eqref{eq:target_state}, we set the number of pairs $N = 8, 10, 15, 20$ and apply our searching algorithm to find the optimal EW operator with maximal noise tolerance. Part of the numerical results are shown in Table \ref{tab:result1}.

\begin{table*}[htbp!]
\caption{Optimal $F_{\mathrm{min}}$ under different LMCs for target states of Eq.~\eqref{eq:target_state} with different numbers of pairs $N$. When $N$ is large the tolerable fidelity is given by $1-p_{max}\approx13/16=81.25\%$ in Eq.~\eqref{eq:maxtolerablep}.}
\label{tab:result1}
\centering
\begin{tabular}{p{35pt}p{35pt}p{32pt}p{32pt}p{32pt}p{32pt}p{32pt}p{32pt}}
\hline
\hline
$N=8$ & LMC & 18 & 54 & 162 & 486 & 1458 & 4374 \\
& $F_{\mathrm{min}}$ & 94.6\% & 92.5\% & 90.0\% & 84.2\% & 81.8\% & 81.2\%\\
\hline
$N=10$ & LMC & 18 & 54 & 162 & 486 & 1458 & 4374 \\
& $F_{\mathrm{min}}$ & 95.8\% & 94.2\% & 92.8\% & 90.3\% & 84.6\% & 82.0\%\\
\hline
$N=15$ & LMC & $2\times 3^{2}$ & $2\times 3^{4}$ & $2\times 3^{6}$ & $2\times 3^{8}$ & $2\times 3^{10}$ & $2\times 3^{12}$ \\
& $F_{\mathrm{min}}$ & 97.3\% & 95.5\% & 93.4\% & 87.4\% & 81.8\% & 81.3\%\\
\hline
$N=20$ & LMC & $2\times 3^{2}$ & $2\times 3^{5}$ & $2\times 3^{8}$ & $2\times 3^{11}$ & $2\times 3^{14}$ & $2\times 3^{17}$ \\
& $F_{\mathrm{min}}$ & 98.0\% & 95.9\% & 93.6\% & 85.0\% & 81.3\% & 81.3\%\\
\hline
\hline
\end{tabular}
\end{table*}

Furthermore, we show the changes in the relationship between LMC and $F_\mathrm{min}$ in Fig.~\ref{Fig:result}. As shown in the figure, the larger LMC always brings better noise tolerance, i.e., the lower requirements for the experimental fidelity. However, when the LMC reaches a value that is large enough, the noise tolerance would not improve a lot anymore. For instance, for the $N = 10$ target state one can find when the LMC is larger than 1458, $F_{\mathrm{min}}$ only decreases a little. In this case, one may not be worth paying a lot of experimental resources for minor detection improvements. In general, there exists a trade-off between detection robustness and experimental efficiency. One can always find an appropriate and efficient EW operator under different noise conditions.

\begin{figure}[htbp!]
\centering
\includegraphics[width=0.5\textwidth]{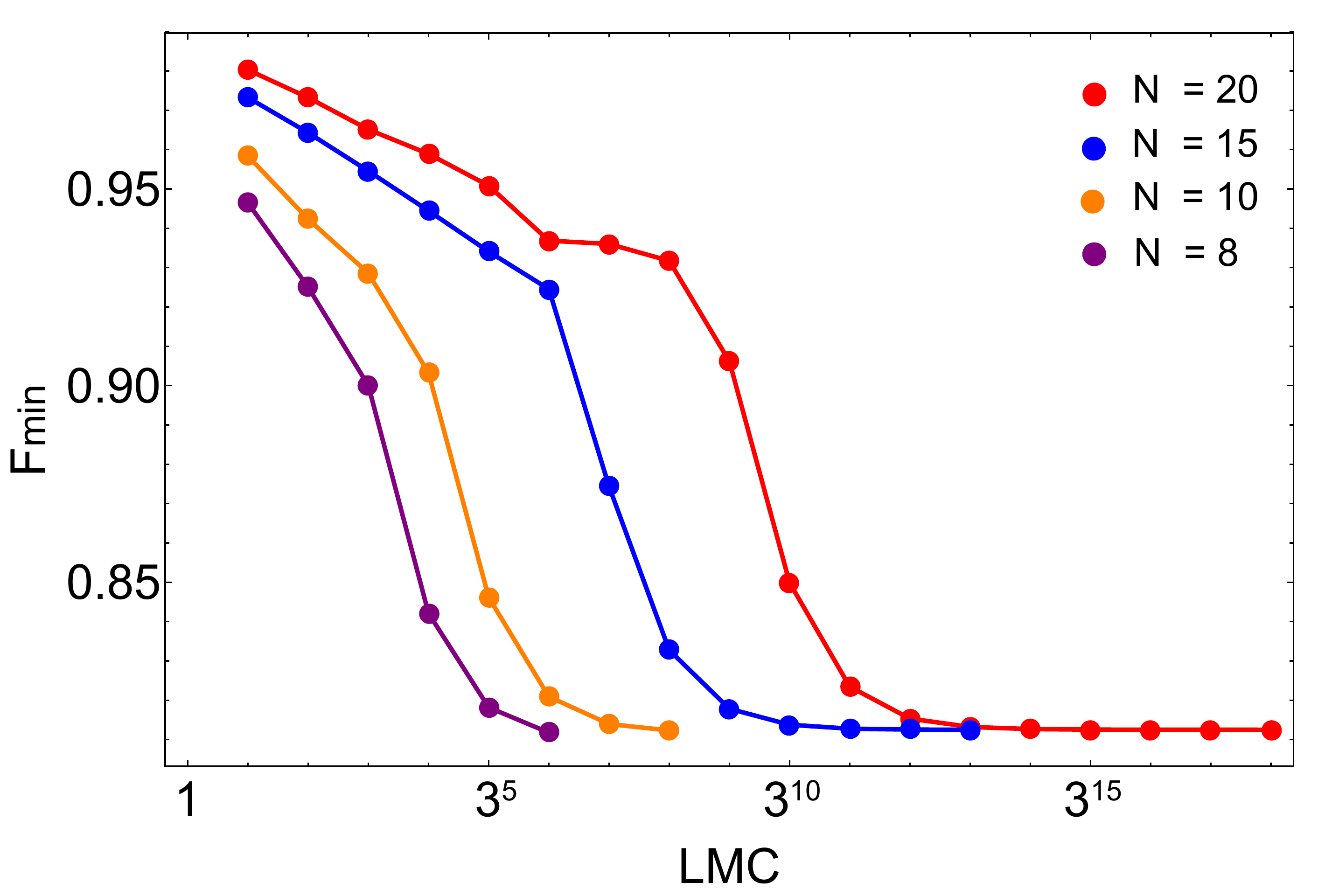}
\caption{Optimal EWs obtained by the numerical searching algorithm for $N = 8, 10, 15, 20$ entangled states. The minimum requirements for the experimental fidelity under different cases are plotted versus LMC. Note that $F_{\mathrm{min}}$ decreases more and more slowly when LMC increases, and finally reaches the lower bound $1 - p_{max}$ [see Eq. ~\eqref{eq:maxtolerablep}], which approaches $81.25\%$ when $N$ is large.}\label{Fig:result}
\end{figure}

\section{Conclusion}\label{Sec:final}
In this paper, we first propose an efficient framework to generate large-scale genuine multipartite entanglement. The generation procedures contain: (1) prepare $N$ singlets; (2) apply the $\sqrt{\mathrm{SWAP}}$ gates. Our generation scheme is simple and efficient to apply in most of the experimental setups. Second, we generalize the stabilizer formalism to analyze the stabilizer structures of the target state, and then design an EW method to verify the genuine multipartite entanglement. The key point of our detection method is to find an appropriate EW operator for the target state. We design a searching algorithm to construct the optimal EW operator under different noise conditions. Finally, we discuss the trade-off between detection robustness and experimental efficiency. Our analytical analysis and the numerical algorithm are generic and can be applied to other entangled states.

The algorithm here can be improved further. To ensure the efficiency of the algorithm, the searching method we use is to find the best solution in a reasonable time. However, it is not guaranteed that the solution is the global optimum of the optimization. There may exist an efficient algorithm to find the optimal EW operator. Besides, our searching region is restricted to the projector-product subsets. Thus, the constructed EW operator is evaluated as the `optimal' in terms of the projector region. As a result, it is interesting to construct the optimal EW operator beyond the stabilizer formalism. One possible way is to employ some modern optimization techniques, such as machine learning, and we leave it for further research.

\section{Acknowledgments}
We acknowledge Q.~Zhao, Z.~Yuan, and P.~Zeng for the insightful discussions. This work was supported by the National Natural Science Foundation of China Grant No.~11875173, the National Key Research and Development Program of China Grants No.~2019QY0702 and No.~2017YFA0303903, and the Zhongguancun Haihua Institute for Frontier Information Technology. Y.Z.~is supported by the National Research Foundation (NRF), Singapore, under its NRFF Fellow program (Award No.~NRF-NRFF2016-02), the Quantum Engineering Program Grant QEP-SF3, the Singapore Ministry of Education Tier 1 Grants No.~MOE2017-T1-002-043, and No.~FQXi-RFP-1809 from the Foundational Questions Institute and Fetzer Franklin Fund (a donor-advised fund of Silicon Valley Community Foundation).

\bibliographystyle{apsrev}

\bibliography{iSwapEW}

\end{document}